\documentclass[12pt]{article}
\usepackage[breaklinks=true]{hyperref}
\usepackage{bm}
\usepackage[utf8]{inputenc}
\usepackage{hyperref} 
\hypersetup{colorlinks=true, linkcolor=black, citecolor=black, urlcolor=blue!80!black} 
\usepackage{amsfonts,amsmath,amsxtra,amsthm}
\usepackage{amssymb}
\usepackage{lscape}
\usepackage[normalem]{ulem}
\usepackage{breqn}
\usepackage[russian,english]{babel}
\usepackage{xcolor}
\usepackage{tikz, pgfplots}
\usepackage{tcolorbox}
\usepackage{mathrsfs}
\usepackage{cancel}

\numberwithin{equation}{section}

\renewcommand{\epsilon}{\varepsilon}

\def\a{\alpha}

\def\b{\beta}

\def\beq{\begin{equation}}
	\def\eeq{\end{equation}}
\def\beqq{\begin{equation*}}
	\def\eeqq{\end{equation*}}

\def\bs{\begin{split}}
	\def\es{\end{split}}

\def\bl{{\boldsymbol{\lambda}}}

\def\bmm{{\boldsymbol{\mu}}}

\def\bt{\boldsymbol{t}}

\def\bx{\boldsymbol{x}}
\def\by{\boldsymbol{y}}
\def\bz{\boldsymbol{z}}

\def\bbx{\underline{\boldsymbol{x}}}
\def\bby{\underline{\boldsymbol{y}}}

\def\bbz{\underline{\boldsymbol{z}}}

\def\const{{2\pi\imath}}

\def\Im{\operatorname{Im}}
\def\cg{\mathit{g}}
\def\g{\gamma}

\def\gl{\mathfrak{gl}}

\def\i{\imath}

\def\K{{K}}

\def\l{\lambda}

\def\la{\lambda}

\def\R{\mathbb{R}}
\def\Re{\mathrm{Re}\,}

\def\ve{\varepsilon}

	\newtheorem{lemma}{Lemma}[section]
	\newtheorem{proposition}{Proposition}[section]
	\newtheorem*{proposition*}{Proposition}
	\newtheorem{corollary}{Corollary}[section]
	\newtheorem*{theorem*}{Theorem}
	\newtheorem{theorem}{Theorem}[section]

	\newcommand{\rf}[1]{(\ref{#1})}

	\def\rPsi{\Psi^{}}   
	\def\rQ{Q^{}} 
	 
	\def\rL{\Lambda^{}}
	
	\def\rK{K^{}} \def\ccK{\hat{K}}
	\def\rmu{\mu^{}}

	\renewcommand{\iota}{\imath}
	
\begin{document}

			\begin{center}
				{\bf \large On Borodin - Gorin presentation of Heckman--Opdam $\gl_n$ hypergeometric function \\[4pt] }
				
				\vspace{0.4cm}
				
				{ L. Cherepanov$^{\circ}$,  S. Khoroshkin$^{\circ}$}
				
				\vspace{0.4cm}
				
				{\small \it					
			$^\circ$National Research University Higher School of Economics, \\
					Myasnitskaya 20, Moscow, 101000, Russia \vspace{0.2cm}	}
				
			\end{center}
			
			\begin{abstract} 
				We prove that Borodin - Gorin presentation for Heckman--Opdam $\gl_n$ hypergeometric function by  interlacing integrals is equivalent to two other integral presentations studied in  https://arxiv.org/abs/2508.18864. We also consider related Baxter operators and the corresponding limiting procedure from the Ruijsenaars hyperbolic model

			\end{abstract}
		\section{Introduction} 
		In \cite{BCDK} we proved the coincidence of two dual integral presentation of Heckman - Opdam hypergeometric function, obtained earlier in \cite{HR0} and  \cite{KK1,KK}. They can be described as follows. Set \cite{HR0}
		\begin{multline} \label{Psi-E}
			\Psi_{\bl_n}(\bx_n) = \prod_{k = 1}^{n - 1} \frac{(2\pi \Gamma(g))^k}{k!} \, \int_{\mathbb{R}^{n(n - 1)/2}} \exp \biggl( 2\pi \imath \sum_{k = 1}^n \lambda_k \biggl( \sum_{i = 1}^{k} y_{ki} - \sum_{i = 1}^{k - 1} y_{k - 1, i} \biggr) \biggr) \\[6pt]
			\times  \prod_{k = 1}^{n - 1} \frac{\prod_{1 \leq i < j \leq k} \, [4 \sh^2 \pi (y_{ki} - y_{kj})]^g}{ \prod_{i = 1}^{k + 1} \prod_{j = 1}^{k} \, [2\ch \pi (y_{k+ 1, i} - y_{kj})]^g }  \, \prod_{k = 1}^{n - 1} \prod_{i = 1}^k dy_{ki},
		\end{multline}
		where we denote $\bx_n = (x_1, \dots, x_n)$ and identify $y_{ki} \equiv y_{k,i}$, $y_{ni} \equiv x_i$. This expression is equal to \cite{KK1,KK}
		\begin{multline} \label{Psi-MB}
			{\Psi}_{\bl_n}(\bx_n) = \prod_{k = 1}^{n - 1} \frac{(2\pi \Gamma(g))^{-k}}{k!} \, \int_{\mathbb{R}^{n(n - 1)/2}} \exp \biggl( 2\pi \imath \sum_{k = 1}^n x_k \biggl( \sum_{i = 1}^{k} \gamma_{ki} - \sum_{i = 1}^{k - 1} \gamma_{k - 1, i} \biggr) \biggr) \\[6pt]
			\times  \prod_{k = 1}^{n - 1} \frac{ \prod_{i = 1}^{k + 1} \prod_{j = 1}^{k} \, \Gamma \bigl( \imath \gamma_{k + 1, i} - \imath  \gamma_{kj} + \frac{g}{2} \bigr) \, \Gamma \bigl( \imath \gamma_{k j} - \imath  \gamma_{k + 1, i} + \frac{g}{2} \bigr) }{\prod_{1 \leq i \not= j \leq k} \, \Gamma(\imath \gamma_{k i} - \imath \gamma_{k j}) \Gamma(\imath \gamma_{kj} - \imath \gamma_{ki} + g)} \, \prod_{k = 1}^{n - 1} \prod_{i = 1}^k d\gamma_{ki},
		\end{multline}
		where we identify $\gamma_{ni} \equiv \lambda_i$. Then, as was proved in \cite{BCDK} the function 
		\begin{align} \label{00} 
		\mathcal{F}_{\bl_n}(\bx_n)= { \prod_{k=1}^n \frac{\Gamma(kg)}{\Gamma(g)}}\, \prod_{1\leq i\not=j \leq n}\Gamma^{-1}(\imath(\l_i-\l_j)+g) \, {\Psi}_{\bl_n}(\bx_n).
		\end{align}
		is $\gl_n$ Heckman - Opdam hypergeometric function for all $g>0$. 
		
		A.Borodin and V.Gorin suggested  another expression for Heckman - Opdam hypergeometric function by means of  interlaced integrals. 
		It can be described as follows. Using vector notation
		\beqq \bx_{n}=(x_1,\ldots, x_n), \qquad \by_m=(y_{m,1},\ldots y_{m,m}) \eeqq
		with the partial order 
		\beqq \by_m\succ\by_{m-1}\eeqq deciphering as
		\beqq y_{m,1}<y_{m-1,1}<y_{m,2}\ldots<y_{m-1,m-1}<y_{m,m}, \eeqq
		set for $\bx_n=\{x_1<x_2<\ldots <x_n\}$ 
		\beq \label{01}\begin{split} \Psi^*_{\bl_n}(\bx_n)= \!\!
		\prod\limits_{1\leq a<b\leq n}\left(\frac{2\pi}{\Gamma(g)} |2\sh\pi(x_a-x_b)|^{1-2g}\right)\!\!\!\!\!\!\!\!
		\int\limits_{\bx_n=\by_n\succ\by_{n-1}\succ\ldots\succ\by_1}\!\!\!\!\!\!\!\!\!\!\! e^{2\pi i\sum\limits_{k=1}^n \l_k(\sum\limits_{i=1}^ky_{k,i}-\sum\limits_{j=1}^{k-1}y_{k-1,j})}                    \ \\ 
		\prod_{k=1}^{n-1}\left(\prod_{1\leq a<b\leq k}|2\sh\pi(y_{k,a}-y_{k,b})|^{2-2g}\prod_{i=1}^k\prod_{j=1}^{k+1}|2\sh\pi(y_{k,i}-y_{k+1,j})|^{g-1}\right)d\by_{k}.\end{split}\eeq 
		In this note, see Theorem \ref{theorem1}, we prove using rational degereration of E.Rains integral identities that
		 \beq \label{02}\Psi^*_{\bl_n}(\bx_n)=  \a_n^{-1} \mathcal{F}_{\bl_n}(\bx_n),\eeq
		 where
		 \beqq \label{Q31}\a_n=\prod_{k=1}^n\frac{\Gamma(kg)}{\Gamma(g)}.\eeqq
		For that we introduce corresponding raising operators, commutative family of new Baxter operators, study local relations between them   and between raising and Baxter  operators, related to the presentations \rf{Psi-E} and \rf{Psi-MB}, and the spectral properties of new Baxter operators, see Theorem \ref{theorem2}.
		
		Besides, we observe the limiting procedure from the hyperbolic system, which relates the results of  \cite{BDKK4} with the present paper.

		\setcounter{equation}{0}
		\section{Hallnas - Ruijsenaars presentation}
	 Both presentations \rf{Psi-E} and \rf{Psi-MB}, as well as their equality, can be obtained by the study of local raising and Baxter operators \cite{BCDK}. In the present paper we deal only with preentation \rf{Psi-E}
		  Set
		 \beq\label{s1}  K(x)=(2\ch \pi x)^{-\cg}, \qquad w(x)=|2\sh \pi x|^{\cg},\qquad \mu(x)=w^2(x)\eeq
		 and, correspondingly, for
		 \beqq \bx_{n}=(x_1,\ldots, x_n), \qquad \by_m=(y_{1},\ldots y_{m}), \eeqq
		 set 
		 \beq\label{s2} \begin{split} K(\bx_n,\by_m)=\prod_{i = 1}^n \prod_{j = 1}^m K(x_i-y_j), \qquad \mu(\bx_n)= \prod_{1\leq i<j \leq n} \mu(x_i-x_j).\end{split}\eeq 
	The raising operator $\rL_{n}(\l)$ is the linear  operator transforming the function of $n-1$ variable to the function of $n$ variables. It is the integral operator
	\beq\label{r14}\begin{split} \left(\rL_{n}(\l)f\right)(\bx_n)=d_{n - 1} \int_{\R^{n-1}}d\by_{n-1} \, \rL(\bx_n,\by_{n-1};\l) f(\by_{n-1})
	\end{split}\eeq
	with the kernel
	\beq\label{r15} \rL(\bx_n,\by_{n-1};\l)= e^{\const \l(\bbx_n-\bby_{n-1})}
	\rK(\bx_n,\by_{n-1})\rmu  (\by_{n-1}), \qquad x_j,y_j \in \R,
	\eeq
	{and the normalizing constant} 
	\begin{equation}\label{dconst}
		d_n =  \frac{(2\pi \Gamma(g))^n}{n!}.
	\end{equation}
	{We denote}
	\beqq \bbx_n=\sum_{k=1}^n x_k,\qquad \bby_{n-1}=\sum_{k=1}^{n-1}y_k .\eeqq 
	According to \rf{Psi-E}, the wave function $\rPsi_{\bl_n}(\bx_n)$ can be described by means of  successive application of raising operators starting from the plane wave
	\beq \label{r16}\Psi_{\bl_n}(\bx_n)=\rL_{n}(\l_n) \, \rL_{n-1}(\l_{n-1})\cdots \rL_{2}(\l_2) \, e^{\const \l_1x_1}.
	\eeq
	The Baxter operators $Q_n(\l)$ transform functions of $n$ variables to the functions of $n$ variables. They are given by similar expressions
	\begin{equation}\label{r9}
		\left( \rQ_{n}(\lambda) f\right) (\bm{x}_n) =  d_n \, \int_{\mathbb{R}^n} d\bm{y}_n \, \rQ(\bm{x}_n, \bm{y}_n; \lambda) f(\bm{y}_n)
	\end{equation}
	with the kernel
	\beq\label{r10} \rQ(\bx_n,\by_{n};\l)= e^{\const \l(\bbx_n-\bby_n)}
	\rK(\bx_n,\by_{n})\rmu  (\by_{n}), \qquad x_j,y_j \in \R.
	\eeq	
	According to \cite{BCDK}, raising operators satisfy  exchange relations
	\begin{align}\label{s19} 
		\Lambda_n(\l)\Lambda_{n-1}(\rho)=\Lambda_n(\rho)\Lambda_{n-1}(\l)
	\end{align}
	under assumption $| \Im (\l - \rho) | < g$.
	 Under the same assumption the Baxter operators commute (see also \cite{H})
	 \beq\label{s18} Q_n(\l)Q_n(\rho)=Q_n(\rho)Q_n(\l).\eeq
	 	Denote by
	 	\begin{equation}\label{s12}
	 		\ccK(\l) = \Gamma \Bigl(\imath\l + \frac{g}{2} \Bigr) \Gamma\Bigl(-\imath\l+\frac{g}{2}\Bigr)
	 	\end{equation}
	the dual kernel function for the Sutherland model. The following local relation was also proved in \cite{BCDK}	
	\beq\label{s20} Q_n(\l)\Lambda_n(\rho)=\ccK(\l-\rho) \, \Lambda_n(\rho)Q_{n-1}(\l)\eeq	
	under assumption 
		 $|\Im(\lambda-\rho)| < \frac{g}{2}$.	
	As a corollary, we see  that the wave function $\rPsi_{\bl_n}(\bx_n)$ diagonalize not only Calogero - Sutherland Hamiltonians ,
	 but also Baxter operators (that was indepedently shown by M.Hallnas \cite{H})
	\beq\label {s11}
	Q_n(\lambda) \, \Psi_{\bm{\lambda}_n}(\bx_n) = \prod^n_{i=1}\ccK(\lambda-\lambda_i) \, \Psi_{\bm{\lambda}_n}(\bx_n).
	\eeq
	
		\setcounter{equation}{0}
	\section{Interlacing  Baxter and raising operators.}

 {\bf 1.} Define operators $Q^*_{n,\pm}(\mu)$ which transform  symmetric functions of $n$ real variables to  symmetric functions of $n$ variables so that for $x_1<x_2\ldots<x_n$
		\beq\label{I0}\begin{split} Q^*_{n,+}(\mu)f(\bx_n)=d_n^*\prod_{i<j}|2\sh\pi(x_i-x_j)|^{-2g+1}\hspace{-1.4cm}\int\limits_ {x_1<y_1<x_2<\ldots<y_n<\infty}\hspace{-1.2cm} d\by_n e^{2\pi i\mu(\bbx_n-\bby_n)}\times\\ \prod_{i<j}|2\sh\pi(y_i-y_j)| \prod_{i,k=1}^n |2\sh\pi(x_i-y_k)|^{g-1} f(\by_n) \end{split}\eeq
		and 
		\beq\label{I1}\begin{split} Q^*_{n,-}(\mu)f(\bx_n)=d_n^*\prod_{i<j}|2\sh\pi(x_i-x_j)|^{-2g+1}\hspace{-1.4cm}\int\limits_ {-\infty<y_1<x_1<\ldots<y_n<x_n}\hspace{-1.2cm} d\by_n e^{2\pi i\mu(\bbx_n-\bby_n)}\times\\ \prod_{i<j}|2\sh\pi(y_i-y_j)| \prod_{i,k=1}^n |2\sh\pi(x_i-y_k)|^{g-1} f(\by_n), \end{split}\eeq
{where}
	 \beq\label{I3} d_n^*=
\left(\frac{2\pi}	{\Gamma(g)}\right)^n.\eeq

{ Note that for real $\mu$ operators $Q^*_{n,\pm}(\mu)$ form a symmetric pair 
 with respect either scalar product
 \beq\label{pair1} \langle f,g\rangle=\int_{\R^n}\bar{f}(\bx_n)g(\bx_n)\tilde{\Delta}_n(\bx_n)^{2g}d\bx_n\eeq
 or bilinear pairing  of symmetric functions
 \beq\label{pair2}  ( f,g)=\int_{\R^n}\bar{f}(-\bx_n)g(\bx_n)\tilde{\Delta}_n(\bx_n)^{2g}d\bx_n,\eeq
 where
 \beq\label{Delta} \tilde{\Delta}_n(\bx_n) = \prod_{1\leq i<j\leq n}|2\sh\pi(x_i-x_j)|,
 	\eeq
 once the corresponding integrals  converge
\beq\label{pair3} \begin{split} \langle Q^*_{n,\pm}(\mu)f,g\rangle= \langle f, Q^*_{n,\mp}(\mu) g\rangle,\qquad ( Q^*_{n,\pm}(\mu)f,g)= ( f, Q^*_{n,\mp}(\mu) g).
\end{split}\eeq

		\begin{proposition}\label{propositionQ} For $\Re g\in (0,1)$ and real $\bl$ and $\bmm$  the wave functions $\Psi_{\bl_n}(\bx_n)$ are in the domain of definition of the operators $Q^*_{n,\pm}(\mu)$ \end{proposition}
		\begin{proof} Since the wave function $\Psi_{\bl_n}(\by_n)$ has no singularities for real $y_k$, in particular, at hyperplanes $y_i=y_j$, the integrable for $\Re g>0$ singularities $y_i=x_k$ do not affect the convergence of the total integral. The asymptotics of the kernel of the operator $Q^*_{n,+}(\mu)$ with respect to the variable $y_n$ tending to infinity is proportional to
			\beqq e^{\pi (-2i\mu+ng-1)y_n}.
			\eeqq
			On the other hand, we have a total bound of the wave function $\Psi_{\bl_n}(\by_n)$ (see \cite{BCDK}[Theorem 2.2])
			\beq \label{I1a} \Psi_{\bl_n}(\bx_n)<C(\bl)e^{-\pi g(\sum_{i<j}|y_i-y_j|)},\eeq
		so whenever the integration variables $y_1,\ldots, y_{n-1}$ vary over a compact domain, the whole integrand is bounded by
		\beqq Ce^{\pi \left((-2i\mu+ng-1)- (n-1)g\right)y_n}=Ce^{\pi \left((-2i\mu+(g-1)\right)y_n} \eeqq
		 which proves the convergence of the integral presenting the action of the operator $Q^*_{n,+}(\mu)$ with additional assumption $\Re g<1$. The bounds for the operator $Q^*_{n,-}(\mu)$ are similar.
		 \end{proof}
		
		Define also raising operator $\Lambda^*_n(\mu)$,  which transforms symmetic functions of $n-1$ real variables to symmetic functions of $n$ real variables
	\beq\label{I2}\begin{split} \Lambda^*_{n}(\mu)f(\bx_n)=d_{n-1}^*\prod_{i<j}|2\sh\pi(x_i-x_j)|^{-2g+1}\hspace{-1.4cm}\int\limits_ {x_1<y_1<\ldots<y_{n-1}<x_n}\hspace{-1.2cm} d\by_{n-1} e^{2\pi i\mu(\bbx_n-\bby_{n-1})}\times\\  \prod_{i<j}|2\sh\pi(y_i-y_j)|\prod_{i=1}^n\prod^{n-1}_{j=1} |2\sh\pi(x_i-y_j)|^{g-1}f(\by_{n-1}).
	\end{split}\eeq
	Contrary to \rf{pair3}, the sole raising operator is symmetric
\beq\label{pair4}  \langle \Lambda^*_{n}(\mu)f,g\rangle= \langle f, \Lambda^*_{n}(\mu) g\rangle,\qquad ( \Lambda^*_{n}(\mu)f,g)= ( f, \Lambda^*_{n,}(\mu) g).
\eeq
\color{black}

	 Since the wave function $\Psi_{\bl_{n-1}}(\bx_{n-1}) $ has no singularities for real $\bx_{n-1}$ and $\bl_{n-1}$, it belongs to the domain of definition of the operator $\Lambda^*_{n}(\mu)$ for real $\mu$.
	\begin{proposition}\label{proposition1} We have the following exchange  relation
	\beq\label{I4}	\Lambda_n(\l)\Lambda_{n-1}^*(\mu)=\Gamma(g\pm i(\l-\mu))\Lambda_n^*(\mu)\Lambda_{n-1}(\l) \eeq
	 under the condition $|\Im(\l-\mu)|<\Re g$.
	\end{proposition}
	\begin{proof} Recall  Rains integral identity in Euler rational degeneration regime \cite{R}
	\beq\label{Rains}\begin{split}
		\frac{\prod_{0\leq i<j\leq n}|a_i-a_j|^{1-\alpha_i-\alpha_j}}
		{\prod_{0\leq i \leq n}\Gamma(\alpha_i)
			\prod_{0\leq j \leq m}|a_i+b_j|^{\alpha_i}}
		&\int\limits_{a_{0}<t_1<\ldots<t_n<a_n}\hspace{-0.8cm} d \bt_n
		\Delta_n(\bt_n)\, \prod_{i=1}^n\, 
		\frac{\prod_{0\leq j\leq n}|a_j-t_i|^{\alpha_j-1}}
		{\prod_{0\leq j \leq m}|t_i+b_j|^{\beta_j}}   = \\
		\frac{\prod_{0\leq i<j\leq m}|b_i-b_j|^{1-\beta_i-\beta_j}}
		{\prod_{0\leq i \leq m}\Gamma(\beta_i)
			\prod_{0\leq j \leq n}|b_i+a_j|^{\beta_i}}
		&\int\limits_{b_{0}<t_1<\ldots<t_m<b_m}\hspace{-0.8cm} d \bt_m
		\Delta_m(\bt_m)\, \prod_{i=1}^m\, 
		\frac{\prod_{0\leq j\leq m}|b_j-t_i|^{\beta_j-1}}
		{\prod_{0\leq j \leq n}|t_i+a_j|^{\alpha_j}},
	\end{split}\eeq
	where \beq\label{I5}\begin{split}&-b_m<-b_{m-1}<\cdots<-b_0<a_0<a_1<\cdots<a_n,\\ 
	&\alpha_0+\cdots+\alpha_n = \beta_0+\cdots+\beta_m,\qquad \Re\a_i>0,\qquad\Re\b_i>0,   \end{split}\eeq
	\begin{align}
		\Delta_n(t) = \prod_{1\leq i<j\leq n}|t_i-t_j|.
	\end{align}	
	Set
	\beqq a_i=e^{2\pi x_i},\qquad b_i=e^{2\pi z_i},\qquad t_i=e^{2\pi y_i}.\eeqq
	In these variables Rains identity looks as 
	\beq\label{Q5}\begin{split}
		\frac{(2\pi)^n\!\!\!\prod\limits_{0\leq i<j\leq n}|2\sh\pi(x_i-x_j)|^{1-\alpha_i-\alpha_j}}
		{\prod\limits_{0\leq i \leq n}\Gamma(\alpha_i)\!
			\prod\limits_{0\leq j \leq m}|2\ch\pi(x_i-z_j)|^{\alpha_i}}\hspace{-0.2cm}
		\int\limits_{x_{0}<y_1<\ldots<y_n<x_n}\hspace{-0.8cm} d \by_n
		\tilde{\Delta}_n(y)\, \prod_{i=1}^n\, 
		\frac{\prod\limits_{0\leq j\leq n}|2\sh\pi(x_j-y_i)|^{\alpha_j-1}}
		{\prod\limits_{0\leq j \leq m}|2\ch\pi(z_j-y_i)|^{\beta_j}}   = \\
		\frac{(2\pi)^m\!\!\!\prod\limits_{0\leq i<j\leq m}|2\sh\pi(z_i-z_j)|^{1-\beta_i-\beta_j}}
		{\prod\limits_{0\leq i \leq m}\Gamma(\beta_i)\!
			\prod\limits_{0\leq j \leq n}|2\ch\pi(z_i-x_j)|^{\beta_i}}\hspace{-0.2cm}
		\int\limits_{z_{0}<y_1<\ldots<y_m<z_m}\hspace{-0.8cm} d \by_m
		\tilde{\Delta}_m(y)\, \prod_{i=1}^m\, 
		\frac{\prod\limits_{0\leq j\leq m}|2\sh\pi(z_j-y_i)|^{\beta_j-1}}
		{\prod\limits_{0\leq j \leq n}|2\ch\pi(y_i-x_j)|^{\alpha_j}},
	\end{split}\eeq
	where $\label{Delta} \tilde{\Delta}_n(y)$ is defined in \rf{Delta}.
		
	Now we  choose $n=m$ in Rains identity \rf{Q5}; take the parameters in natural order
	\beqq x_0<x_1\ldots<x_n,\qquad z_0<z_1\ldots<z_n\eeqq
	and set
	\beq\label{Q6}\begin{split}
		\alpha_i&=g,\qquad i=0,\ldots,n,\qquad\qquad \beta_j=g,\qquad j=1,\ldots,n-1,\\ \beta_0&=g-\rho,\qquad \beta_n=g+\rho,\qquad
		z_0=-Z,\qquad z_n=Z. \end{split}\eeq
	Note that according to \rf{I5},
	\beqq |\Re \rho|<g. \eeqq	
	Then the asymptotics of LHS of \eqref{Q5} when $Z$ tends to infinity are
	\beq \label{Q8}\begin{split}\frac{(2\pi)^n\prod\limits_{0\leq i<j\leq n}|\sh\pi(x_i-x_j)|^{1-2g}}{\Gamma^{n+1}(g)\prod\limits_{i=0}^{ n}\prod\limits_{j=1}^{ n-1}|\ch\pi(x_i-z_j)|^g}\times e^{-(4n+2)g\pi Z}\times &
		\\ \int\limits_{ x_0<y_1<\ldots <y_n<x_n}
		\hspace{-0.8cm} d \by_ne^{2\pi\rho\sum\limits_{j=1}^n y_j}
		&\tilde{\Delta}_n(y)\,  
		\frac{\prod\limits_{i=1}^n\prod\limits_{j=0}^{n}|2\sh\pi(x_j-y_i)|^{g-1}}
		{\prod\limits_{i=1}^n\prod\limits_{j=1}^{n-1}|2\ch\pi(z_j-y_i)|^{g}} .
	\end{split}\eeq
	The asymptotics of RHS of \eqref{Q5} when $Z$ tends to infinity are
	\beq \label{Q9}\begin{split}\frac{(2\pi)^n\prod\limits_{1\leq i<j\leq {n-1}}|\sh\pi(z_i-z_j)|^{1-2g}}{\Gamma^{n-1}(g)\Gamma(g+\rho)\Gamma(g-\rho)\prod\limits_{i=0}^{ n}\prod\limits_{j=1}^{ n-1}|\ch\pi(x_i-z_j)|^g}&\times e^{-(4n+2)g\pi Z+2\pi \rho\big(\sum_ix_i+\sum_jz_j\big)}\times 
		\\ \int\limits_{ -Z<y_1<z_1\ldots<z_{n-1} <y_n<Z}
		\hspace{-0.8cm} d \by_ne^{-2\pi\rho\sum\limits_{j=1}^n y_j}
		&\tilde{\Delta}_n(y)\,  
		\frac{\prod\limits_{i=1}^n\prod\limits_{j=1}^{n-1}|2\sh\pi(z_j-y_i)|^{g-1}}
		{\prod\limits_{i=1}^n\prod\limits_{j=0}^{n}|2\ch\pi(x_j-y_i)|^{g}} .
	\end{split}\eeq 
	The integral in the RHS of \rf{Q8} converges absolutely for any $g>0$ and $|\Re \rho|<g$,  
	thus the integral 
	\beqq   \int\limits_{ -\infty<y_1<z_1\ldots<z_{n-1} <y_n<\infty}
	\hspace{-0.8cm} d \by_ne^{-2\pi\rho\sum\limits_{j=1}^n y_j}
	\tilde{\Delta}_n(y)\,  
	\frac{\prod\limits_{i=1}^n\prod\limits_{j=1}^{n-1}|2\sh\pi(z_j-y_i)|^{g-1}}
	{\prod\limits_{i=1}^n\prod\limits_{j=0}^{n}|2\ch\pi(x_j-y_i)|^{g}} \eeqq
	converges absolutely as well.
	
	Comparing \eqref{Q8} and \eqref{Q9} and setting $\rho=i\gamma$, we come to the equality
	\beqq \label{Q10}\prod\limits_{0\leq i<j\leq n}\!\!|\sh\pi(x_i- x_j)|^{1-2g}\hspace{-0.8cm} 
		\int\limits_{ x_0<y_1<\ldots <y_n<x_n}
		\hspace{-0.8cm} d \by_ne^{2\pi i\gamma(\sum\limits_{j=1}^n y_j-\sum\limits_{j=0}^nx_j)}
		\tilde{\Delta}_n(y)\,  
		\frac{\prod\limits_{i=1}^n\prod\limits_{j=0}^{n}|2\sh\pi(x_j-y_i)|^{g-1}}
		{\prod\limits_{i=1}^n\prod\limits_{j=1}^{n-1}|2\ch(z_j-y_i)|^{g}}=
		\eeqq
	\beq \label{Q10}\begin{split}	
		& 
		\frac{\Gamma^2(g)}{\Gamma(g\pm i\gamma)}{\prod\limits_{1\leq i<j\leq n-1}\!\!|\sh\pi(z_i-z_j)|^{1-2g}}\times\\ 
	&	\int\limits_{ y_1<z_1<\ldots <z_{n-1}<y_n}
		\hspace{-0.8cm} d \by_ne^{2\pi i\gamma(\sum\limits_{j=1}^{n-1} z_j-\sum\limits_{j=1}^ny_j)}
		\tilde{\Delta}_n(y)\,  
		\frac{\prod\limits_{i=1}^n\prod\limits_{j=1}^{n-1}|2\sh\pi(z_j-y_i)|^{g-1}}
		{\prod\limits_{i=1}^n\prod\limits_{j=0}^{n}|2\ch\pi(x_j-y_i)|^{g}}.	 
	\end{split}\eeq

	 It is equivalent to exchange relation \rf{Q8} on the level of equality of kernels of compositions of operators in both sides of equality.
    	\end{proof}
	\begin{proposition}\label{proposition1a} For real $\bx_n$, $\bl_{n-1}$, and $\mu$
		\beq \label{Q25}\left(\Lambda^*_n(\mu)\Psi_{\bl_{n-1}}\right)(\bx_n)= \prod_{j=1}^{n-1}\Gamma^{-1}(\pm i(\mu-\l_j))
		\Psi_{\mu,\bl_{n-1}}(\bx_n) \eeq
		\end{proposition} 
		\begin{proof}
			The proof is going by induction with the use of Proposition \ref{proposition1}.
			We have 
		\beq \label{Q26}	\left(\Lambda^*_n(\mu)\Psi_{\bl_{n-1}}\right)(\bx_n)=\left(\Lambda^*_n(\mu)\Lambda_{\bl_{n-1}}\Psi_{\bl_{n-2}}\right)(\bx_n)\eeq
		Assuming that we can use Proposition \ref{proposition1} in a presence of $\Psi_{\bl_{n-2}}(\by_{n-2})$ (this requires a check of convergences of the related integrals to apply the Fubini's theorem), we get
		\beq\label{Q27} \left(\Lambda^*_n(\mu)\Lambda_{{n-1}}(\l_n)\Psi_{\bl_{n-2}}\right)(\bx_n)=\Gamma^{-1}(\pm i(\mu-\l_n))\left(\Lambda_{{n}}(\l_n)\Lambda^*_{n-1}(\mu)\Psi_{\bl_{n-2}}\right)(\bx_n).
		\eeq
		By induction assumption, 
		\beqq \left(\Lambda^*_{n-1}(\mu)\Psi_{\bl_{n-2}}\right)(\by_{n-1})=\prod_{j=1}^{n-2}\Gamma^{-1}(\pm i(\mu-\l_j)
		\Psi_{\mu,\bl_{n-2}}(\by_{n-1})\eeqq
		which implies the statement of Proposition due to symmetricity of $\Psi_{\bl_n}(\bx_n)$ with respect to the spectral parameters.
		
		The use of Proposition \ref{proposition1} requires a check of convergences of multiple integrals entering into the definition of both expressions
		\beq\label{Q28} \left(\Lambda^*_n(\mu)\Lambda_{n-1}(\l_{n-1})\Psi_{\bl_{n-2}}\right)(\bx_n) \eeq
		  and
		  \beq \label{Q29} \left(\Lambda_n(\l_{n-1})\Lambda_{n-1}^*(\mu)\Psi_{\bl_{n-2}}\right)(\bx_n) .\eeq
		  For the study of the finite singularities, we will systematically apply the following result.
		  
		  \begin{lemma}\label{le_lst}
		  For $a \in (-1,+\infty)$, we have the bound
		  \beq\begin{split}
		  \int\limits_{x_1<y_1<\ldots <x_n}d\by_{n-1}\prod^n_{i=1}\prod^{n-1}_{j=1}|x_i-y_j|^a\times\prod_{1\leq i<j<n-1}|y_i-y_j| \leq C_n(a)\prod_{1\leq i < j \leq n}|x_i-x_j|^{2a+1}.
		  \end{split}\eeq
		  \end{lemma}
		  The proof will be given in Appendix. 
		  
		   Firstly, consider the  integral \rf{Q28}. It can  be viewed as an integral over the bounded set $x_1<x_2<\ldots<x_n$ with the kernel
		  \beqq |2\sh\pi(x_i-x_j)|^{-2g+1} e^{2\pi i\mu(\bbx_n-\bby_{n-1})}\times\\  \prod_{i<j}|2\sh\pi(y_i-y_j)|\prod_{i=1}^n \prod_{j=1}^{n-1}|2\sh\pi(x_i-y_k)|^{g-1} \Psi_{\bl_{n-1}}(\by_{n-1}).\eeqq
		  The integral has singularities at hyperplanes $y_j=x_{j}$ and $y_j=x_{j+1}$ , which are integrable for $g>0$. Since the function $\Psi_{\bl_{n-1}}(\by_{n-1})$ is also given by convergent integral \cite{BDKK2}, the whole expression \rf{Q28} is presented by a convergent integral as well.		  
		  
		   In the integral \rf{Q29} we also meet singularities at hyperplanes, when $0<g\leq 1$, and at infinity for any $g>0$. The integrand looks as
		   \beqq \begin{split}e^{2\pi i(\l_{n-1}(\bbx_n-\bby_{n-1})+\mu(\bby_{n-1}-\bbz_{n-2}))} \prod_{i=1}^n\prod_{j=1}^{n-1}(2\ch\pi(x_i-y_j))^{-g}\prod_{i<j}|
		   2\sh\pi(y_i-y_j)|^{2g}\times \\ \prod_{i<j}|
		   		   2\sh\pi(y_i-y_j)|^{-2g+1}
		   \prod_{i=1}^{n-1}\prod_{j=1}^{n-2}|2\sh\pi(y_i-z_j)|^{g-1}\prod_{i<j}
		   		   |2\sh\pi(z_i-z_j)|{}\Psi_{\bl_{n-2}}(\bz_{n-2})\end{split}\eeqq	   
		   	Now the integration variables $y_j$ vary over $\R^{n-1}$, and we may assume that they are ordered $y_1<y_2<\ldots <y_{n-1}$,  while the variables $z_j$ respect the	interlacing condition
		   	\beqq y_1<z_1<\ldots <z_{n-2}<y_{n-1}. \eeqq   
		   	Contrary to \rf{Q28}, the integral \rf{Q29} for $0<\Re g<1$ may have multiple singularities. Each finite singularity locus is a product of strata \rf{Q28a}
		   	 \beq\label{Q28a} \begin{split}& (I)\ \  y_i=z_i=\ldots= z_j=y_{j+1},\qquad \text{or}\qquad(II)\ \ y_i=z_i=\ldots=y_j =z_j \\& (III)\ \  z_i=y_{i+1}=\ldots =y_j=z_j\qquad \text{or}\qquad  (IV)\ \ z_i=y_{i+1}=\ldots= z_j=y_{j+1}.\end{split}\eeq
		   		   	 We can read the statement of Lemma \ref{le_lst} as a set of  inequalities
		   	\beq\label{Q28b}\begin{split}
		   			  \int\limits_{y_1<z_1<\ldots<y_m}d\bz_{m-1}\prod_{1\leq i < j \leq n}|y_i-y_j|^{1-2g}\prod^m_{i=1}\prod^{m-1}_{j=1}|y_i-z_j|^{g-1}\prod_{1\leq i<j<m-1}|z_i-z_j| \leq C.
		   			  \end{split}\eeq

		\color{black}

		   	Due to the known bound \rf{I1a} on the wave function $\Psi_{\bl_{n-2}}(\bz_{n-2})$ and equivalences  $\sh x\sim x$, $x\to 0$, it means that 
		   	all multiple singularities 
		   	 	are integrable and one can input them into an open set of fixed volume in the corresponding linear subspace of the integration variables $y_i,\ldots y_j$ and $z_l,\ldots z_k$ such that the integral over this  set will be bounded by  a fixed single number.

		   	 	 Then inside these sets we can replace the functions $|2\sh \pi(y_i-z_j)|^{g-1}$, which tend there to infinity, by  $C \exp \pi(g-1)|y_i-z_j|$ with proper  constants without changing total integral. After that we are in a position of the integral without finite singularities and with the same bounds at infinity.
		   	 	 
	Here we should note that for this goal inequality \rf{Q28b} can be used not only for singularity $(I)$ but for others as well just by adding new integration variable inside the chosen small vicinity $|y_i-y_j|<\ve$.  For instance,
		   	 	                   \beq\begin{split}
		   	 	                   y_i<z_i<\ldots<y_j<z_j \qquad \text{becomes} \qquad y_i<z_i<\ldots<y_j<z_j<\tilde{y}_{j+1},
		   	 	                   \end{split}\eeq
where $|\tilde{y}_{j+1}-y_j|=\ve$.
   
		    Then the integrand of \rf{Q29} modified by the procedure above can be bounded by the exponent of
		    \beq \begin{split} \label{Q29.2} \pi(-g\sum_{i=1}^n\sum_{j=1}^{n-1}|x_i-y_j| +\sum_{i<j}|y_i-y_j|+(g-1)\sum_{j=1}^{n-1}\sum_{k=1}^{n-2}|y_j-z^{(n-2)}_k|+ \\
		    + \sum_{1\leq i<j\leq n-2}|z^{(n-2)}_i-z^{(n-2)}_j|+
		     g S_{n-2}(\bz_{n-2},\ldots \bz_1)).\end{split}\eeq

			Here \beqq\by_{n-1}=(y_1,\ldots,y_{n-1}),\qquad \bz_{n-2}=(z_1^{(n-2)},\ldots, z_{n-2}^{(n-2)}),~\ldots,~ \bz_1=z_1^{(1)}\eeqq
			are integration variables with the only restriction
			\beq\label{Q30} y_1<z_1^{(n-2)}< y_2<\ldots< z_{n-2}^{(n-2)}<y_{n-1},\eeq
			where $ x_1,\ldots,x_n$  are external parameters 
			and the function $S_m (\bz_m,\bz_{n-1}.\ldots \bz_1)$ is defined by the recurrence relation
			\beq\label{Q40}\begin{split} & S_m (\bz_m,\bz_{m-1}.\ldots \bz_1)=\\ -  \sum_{i=1}^m\sum_{j=1}^{m-1} |z_i^{(m)}-z_j^{(m-1)}|+
			\sum_{\substack{i,j=1\\i\not=j}}^{m-1}|z_i^{(m-1)}-z_j^{(m-1)}|+& S_{m-1} (\bz_{m-1},\ldots \bz_1).\end{split}\eeq

   The inequality \rf{Q30} implies the relation \cite[Lemma 3.6]{R} 
    \beq\label{Q34}\begin{split}
    \sum_{1\leq i<j\leq n-1}|y_i-y_j| - \sum^{n-1}_{i=1}\sum^{n-2}_{j=1}|y_i-z^{(n-2)}_j| +\sum_{1\leq i<j\leq n-2}|z^{(n-2)}_i-z^{(n-2)}_j| = 0.
    \end{split}\eeq 
     Multiply \rf{Q34} by $2g-1$ and add to \rf{Q29.2}. Then \rf{Q29.2} turns to
\beq \begin{split}\pi g  S_n(\bx_n,\bz_{n-1}, \bz_{n-2},\ldots, \bz_1).
	\end{split}\eeq

	 We now refer to \cite{BDKK2}, where the following bound for this function was obtained
	  \beqq S_n(\bx_n,\bz_{n-1}, \bz_{n-2},\ldots, \bz_1)<C(\bx_n)- \ve \sum_{i=1}^{n-1}\sum_{j=1}^i |z_j^{(i)}|\eeqq
	  for some specific $\ve>0$. This is sufficient for the proof of convergency.
\end{proof}	
	
\begin{theorem}\label{theorem1}
	\beq \label{Q33}\Psi^*_{\bl_n}(\bx_n)=\prod_{k\not=j}\Gamma^{-1}(i(\l_j-\l_k))\Psi_{\bl_n}(\bx_n)\eeq
	\end{theorem}
	\begin{proof} This is a direct corollary of the inductive use of the proposition above since	
		\beq \label{Q24}\Psi^*_{\bl_n}(\bx_n)=\Lambda^*_{n}(\l_n) \, \Lambda^*_{n-1}(\l_{n-1})\cdots \Lambda^*_{2}(\l_2) \, e^{2\pi i \l_1x_1}.
		\eeq
		\end{proof}
		Using the relation between the wave function $\Psi_{\bl_n}(\bx_n)$ and Heckman - Opdam hypergeometric function $\mathcal{F}_{\bl_n}(\bx_n)$ established in \cite{BCDK}, we arrive to the following corollary.
		\begin{corollary}\label{corollary1}
			\beq \label{02}\Psi^*_{\bl_n}(\bx_n)=  \prod_{k=1}^n\frac{\Gamma(g)}{\Gamma(kg)} \mathcal{F}_{\bl_n}(\bx_n)\eeq
		\end{corollary}
		\medskip

{\bf 2.} Let us turn to Baxter operators $Q^*_{n,\pm}(\mu)$.
	\begin{proposition}\label{proposition2}For $\Re g>0$ under the condition $|\Im (\l-\mu)|<\frac{1}{2}$ 
		\beq\label{I55} Q_n(\l)Q^*_{n,\pm}(\mu)=Q^*_{n,\pm}(\mu)Q_n(\la) \eeq
		on the level of kernels corresponding to compositions of operators.
		\end{proposition}
	\begin{proof}
		Consider Rains identity \rf{Q5} for $m=n$ and reenumerate  the variables $x_i$ and $z_i$ so that the integrations regions become
		\beqq x_1<y_1<\ldots <x_n<y_n<x_0\qquad \text{and}\qquad z_0<y_1<z_1\ldots <y_n<z_n.\eeqq
		 Choose
		\beq\begin{split}\label{Q7} \alpha_i&=\beta_i=g,\qquad i=1,\ldots, n, \\  \alpha_0&=\beta_0=\rho,\qquad \Re\rho>0,\\ 
			x_0&=X_0,\qquad z_0=Y_0=-X_0,\end{split}\eeq
		and tend $X_0\to+\infty$. Then the asymptotics of the LHS of \rf{Q5} are
		\beq \label{Q1}\begin{split}\frac{(2\pi)^n\prod\limits_{1\leq i<j\leq n}|2\sh\pi(x_i-x_j)|^{1-2g}}{\Gamma(\rho)\Gamma(g)^n\prod\limits_{i,j=1}^{ n}|2\ch\pi(x_i-z_j)|^g}
			\times e^{-2\pi n X_0(\rho+g)-2\pi X_0\rho}\times e^{\sum_{j=1}^n\pi((\rho-1)x_j+\rho z_j)}\times \\ \int\limits_{ x_1<y_1<\ldots <x_n<y_n<X_0}
			\hspace{-0.8cm} d \by_n e^{\pi(-2\rho+1)\sum\limits_{j=1}^n y_j}
			\tilde{\Delta}_n(y)\, \prod_{i,j=1}^n\, 
			\frac{|2\sh\pi(x_j-y_i)|^{g-1}}
			{|2\ch\pi(z_j-y_i)|^{g}}. 
		\end{split}\eeq
		Analogously, the asymptotics of the RHS are
		\beq \label{Q2}\begin{split}\frac{(2\pi)^n\prod\limits_{1\leq i<j\leq n}|2\sh\pi(z_i-z_j)|^{1-2g}}{\Gamma(\rho)\Gamma(g)^n\prod\limits_{i,j=1}^{ n}|2\ch\pi(x_i-z_j)|^g}
			\times e^{-2\pi n X_0(\rho+g)-2\pi X_0\rho}\times e^{\sum_{j=1}^n\pi(-(\rho-1)z_j-\rho x_j)}\times \\ \int\limits_{ -X_0<y_1<z_1<\ldots <y_n<z_n}
			\hspace{-0.8cm} d \by_n e^{\pi(2\rho-1)\sum\limits_{j=1}^n y_j}
			\tilde{\Delta}_n(y)\, \prod_{i,j=1}^n\, 
			\frac{|2\sh\pi(z_j-y_i)|^{g-1}}
			{|2\ch\pi(x_j-y_i)|^{g}}. 
		\end{split}\eeq
		Both limiting integrals
		\beqq  \int\limits_{ x_1<y_1<\ldots <x_n<y_n<\infty}
		\hspace{-0.8cm} d \by_n e^{\pi(-2\rho+1)\sum\limits_{j=1}^n y_j}
		\tilde{\Delta}_n(y)\, \prod_{i,j=1}^n\, 
		\frac{|2\sh\pi(x_j-y_i)|^{g-1}}
		{|2\ch\pi(z_j-y_i)|^{g}} \eeqq
		and
		\beqq
		\int\limits_{ -\infty<y_1<z_1<\ldots <y_n<z_n}
		\hspace{-0.8cm} d \by_n e^{\pi(2\rho-1)\sum\limits_{j=1}^n y_j}
		\tilde{\Delta}_n(y)\, \prod_{i,j=1}^n\, 
		\frac{|2\sh\pi(z_j-y_i)|^{g-1}}
		{|2\ch\pi(x_j-y_i)|^{g}} \eeqq
		 converge when $\Re \rho>0$.
		Denote
		\beqq \gamma=\mu-\lambda,\qquad 2\rho-1=2i\gamma,\qquad \rho=\frac{1}{2}+{i\gamma}.\eeqq
		Then  \eqref{Q1} and \eqref{Q2} imply the equality
		\beq\label{I5b} Q_n(\l)Q^*_{n,+}(\mu)=Q^*_{n,+}(\mu)Q_n(\la) \eeq
		  in a form
		\beq\label{Q4}\begin{split}
		&{\prod\limits_{1\leq i<j\leq n}|2\sh\pi(z_i-z_j)|^{1-2g}}\hspace{-1.2cm}
			\int\limits_{ -\infty<y_1<z_1<\ldots <y_n<z_n}
			\hspace{-0.9cm} d \by_n e^{2\pi i\gamma\sum\limits_{j=1}^n(y_j- z_j)}\!\!
			\tilde{\Delta}_n(y)
			\prod_{i,j=1}^n 
			\frac{|2\sh\pi(z_j-y_i)|^{g-1}}
			{|2\ch\pi(x_j-y_i)|^{g}}=\\
			&{\prod\limits_{1\leq i<j\leq n}|2\sh\pi(x_i-x_j)|^{1-2g}}\hspace{-1.2cm}
			\int\limits_{ x_1<y_1<\ldots <x_n<y_n<\infty}
			\hspace{-0.9cm} d \by_n e^{2\pi i\gamma\sum\limits_{j=1}^n(x_j- y_j)}\!\!
			\tilde{\Delta}_n(y)\, \prod_{i,j=1}^n 
			\frac{|2\sh\pi(x_j-y_i)|^{g-1}}
			{|2\ch\pi(z_j-y_i)|^{g}}. \end{split}
		\eeq
		 On the other hand, the equality \rf{Q4} could be also read as the condition on kernels for the relation
		 \beq\label{I5c} Q_n(\mu)Q_{n,-}^*(\l)=Q_{n,-}^*(\l)Q_n(\mu) \eeq
		 which implies the second statement of the Proposition.
	\end{proof}

	\begin{proposition}\label{proposition03} Under the condition
		\beq \label{Q22}\Re( i(\l-\mu)+(1-g)/2)>0 \eeq we have the exchange relation
		\beq \label{Q23}Q_{n,+}^*(\l)\Lambda_n(\mu)=
		\frac{\Gamma\left(\frac{1-g}{2}+i(\l-\mu)\right)}{\Gamma\left(\frac{1+g}{2}+i(\l-\mu)\right)}\Lambda_n(\mu) Q_{n-1,+}^*(\l).
		\eeq
	\end{proposition}
	 Under the same conditions we have similar exchange relation for another Baxter operator
	 \beq \label{Q23.5}Q_{n,-}^*(\l)\Lambda_n(\mu)=
	 		\frac{\Gamma\left(\frac{1-g}{2}-i(\l-\mu)\right)}{\Gamma\left(\frac{1+g}{2}-i(\l-\mu)\right)}\Lambda_n(\mu) Q_{n-1,-}^*(\l).
	 		\eeq	 
	\begin{proof} Multiply the equality \rf{Q4} by $e^{\pi n gz_n}$ and tend $z_n\to+\infty$.
		The integral in the RHS  of \rf{Q4} then have a pointwise limit equal to 
		\beqq {\prod\limits_{1\leq i<j\leq n}\!\!\!|2\sh\pi(x_i-x_j)|^{1-2g}}\hspace{-1.2cm}
		\int\limits_{ x_1<y_1<\ldots <x_n<y_n<\infty}
		\hspace{-0.9cm} d \by_n e^{\pi \sum\limits_{j=1}^n(2i\gamma x_j- (2i\gamma-g) y_j)}\!\!
		\tilde{\Delta}_n(y)\, \prod_{i=1}^n 
		\frac{\prod_{j=1}^{n}|2\sh\pi(x_j-y_i)|^{g-1}}
		{\prod_{j=1}^{n-1}|2\ch\pi(z_j-y_i)|^{g}}.\eeqq
		As $y_n\to+\infty$, the integrand has asymptotics equal to
		\beqq e^{\pi(2g-1-2i\gamma) y_n},\eeqq
		and thus converges once
		\beq\label{Q20} \Re (2i\gamma+1-2g) >0. \eeq
		Thus under this condition we have the corresponding limit of integrals.
		
		In the LHS of \rf{Q4} we first make change of variables
		\beqq y_n=z_n+t \eeqq
		and then take a limit. Again pointwisely, we get
		 	the limit 
		 \beq\label{Q21}\begin{split} {\prod\limits_{1\leq i<j\leq n-1}\!\!\!|2\sh\pi(z_i-z_j)|^{1-2g}}\hspace{-1.2cm}
		 \int\limits_{\substack{y_1<z_1<\ldots <y_{n-1}<z_{n-1}\\t<0}}\hspace{-1cm}d\by_{n-1}
		 e^{2\pi i\gamma\sum\limits_{j=1}^{n-1}(y_j- z_j)+\pi g(\sum\limits_{i=1}^nx_i+\sum\limits_{j=1}^{n-1}(z_j-y_j))}\!\!
		 \tilde{\Delta}_{n-1}(y)\times\\
		 \prod_{i=1}^{n-1}	\frac{\prod_{j=1}^{n -1}|2\sh\pi(z_j-y_i)|^{g-1}}
		 {\prod_{j=1}^{n}|2\ch\pi(x_j-y_i)|^{g}}\times e^{(2i\g-g)\pi t}|2\sh\pi t|^{g-1}dt.
		\end{split}\eeq
		The asymptotics of the integrand with respect to the cases $y_1\to-\infty$ and $t\to-\infty$ are
		\beqq e^{\pi(2i\g+1)y_1} \qquad \text{and} \qquad e^{\pi(2i\g+1-2g)t}, \qquad \text{correspondingly.}\eeqq
		Thus, under the condition \rf{Q20} it converges, and we can perform the limit of the identity of the integrals.
		The limiting integral \rf{Q21} decouples into the product of $n-1$ dimensional integral 
		\beqq \int\limits_{y_1<z_1<\ldots <y_{n-1}<z_{n-1}}\hspace{-1cm}d\by_{n-1}
		e^{2\pi i\gamma\sum\limits_{j=1}^{n-1}(y_j- z_j)+\pi g(\sum\limits_{i=1}^nx_i+\sum\limits_{j=1}^{n-1}(z_j-y_j))}\!\!
		\tilde{\Delta}_{n-1}(y)\times
		\prod_{i=1}^{n-1}	\frac{\prod_{j=1}^{n-1}|2\sh\pi(z_j-y_i)|^{g-1}}
		{\prod_{j=1}^{n}|2\ch\pi(x_j-y_i)|^{g}}\eeqq
		and one-dimensional integral
		\beqq \int\limits_{t<0} e^{(2i\g-g)\pi t}|2\sh\pi t|^{g-1}\,  dt=\frac{\Gamma(g)\Gamma(i\gamma-g+\frac{1}{2})}{2\pi \Gamma(i\gamma+\frac{1}{2})}.\eeqq 
		They are convergent under conditions \rf{Q22}.
		After change of parameters
		\beqq i\gamma= i(\l-\mu)+\frac{g}{2} \eeqq
		we finally arrive to the relation \rf{Q23} under restriction \rf{Q22}, taking into account
		 normalizing constants \rf{I3} and \rf{dconst}.
			\end{proof}
			
			\begin{theorem}\label{theorem2} For $ {0<\Re g< 1}$ 
				\beq \label{Q35}Q^*_{n,\pm}(\mu) \Psi_{\bl_n}(
					\bx_n)=
					\prod_{j=1}^n\frac{\Gamma\left(\frac{1-g}{2}\mp i(\l_j-\mu)\right)}{\Gamma\left(\frac{1+g}{2}\mp i(\l_j-\mu)\right)}\Psi_{\bl_n(
						\bx_n)}.
					\eeq
			\end{theorem}
			\begin{proof}
			Similar  to the proof of Theorem \ref{theorem1} and Proposition \ref{proposition1a}, the statement follows by inductive use of Proposition \ref{proposition03} and  requires a check of convergences of multiple integrals entering into the definition of both expressions
					\beq\label{Q36} \left(Q^*_{n,+}(\mu)\Lambda_{n}(\l_{n-1})\Psi_{\bl_{n-1}}\right)(\bx_n) \eeq
					  and
					  \beq \label{Q37} \left(\Lambda_n(\l_{n-1})Q_{n-1,+}^*(\mu)\Psi_{\bl_{n-1}}\right)(\bx_n) .\eeq
			First, we repeat the arguments given in the proof of 	Proposition \ref{proposition1a}	 to show that  all finite singularities are integrable, and we can
			 just pay attention to convergence at infinity. 
			In the most nontrivial case of \rf{Q37},
			we bound the integrand by
			   the exponent of
			  		    \beq \begin{split} \label{Q41} \pi(-g\sum_{i=1}^n\sum_{j=1}^{n-1}|x_i-y_j| +\sum_{i<j}|y_i-y_j|+(g-1)\sum_{j=1}^{n-1}\sum_{k=1}^{n-1}|y_j-z^{(n-1)}_k|+ \\
			  		    + \sum_{1\leq i<j\leq n-1}|z^{(n-1)}_i-z^{(n-1)}_j|+
			  		     g S_{n-1}(\bz_{n-1},\ldots \bz_1))\end{split}\eeq
				are integration variables with the only restriction
						\beq\label{Q38} y_1<z_1^{(n-1)}< y_2<\ldots< z_{n-2}^{(n-1)}<y_{n-1}<z_{n-1}^{(n-1)},\eeq
						$ x_1,\ldots,x_n$  are external parameters 
						and the function $S_m (\bz_m,\bz_{n-1}.\ldots \bz_1)$ is given by the relation \rf{Q40}
					
			   The inequality \rf{Q38} implies the relation (compare \cite[Lemma3.5]{R}) 
			    \beq\label{Q39}\begin{split}
			    \sum_{1\leq i<j\leq n-1}|y_i-y_j| - \sum^{n-1}_{i=1}\sum^{n-1}_{j=1}|y_i-z^{(n-1)}_j| +\sum_{1\leq i<j\leq n-2}|z^{(n-1)}_i-z^{(n-2)}_j| +\\ \sum_{j=1}^{n-1}|z_j^{(n-1)}-y_j|=0.
			    \end{split}\eeq 
			     Multiply \rf{Q39} by $2g-1$ and add to \rf{Q41}. Then \rf{Q41} turns to
			\beqq \begin{split}\pi g  T_n(\bx_n,\by_{n-1}, \bz_{n-1},\ldots, \bz_1)+(2g-1)\sum_{j=1}^{n-1}|z_j^{(n-1)}-y_j|
				\end{split}\eeqq
				where
				\beqq \begin{split}T_n(\bx_n,\by_{n-1}, \bz_{n-1},\ldots, \bz_1)= -\sum_{i=1}^n\sum_{j=1}^{n-1}|x_i-y_j|+\sum_{\substack{i,j=1\\ i\not=j}}^{n-1}|y_i-y_j|\\
				-\sum_{i,j=1}^{n-1}|y_i-z_j^{(n-1)}+	\sum_{\substack{i,j=1\\i\not=j}}^{n-1}|z_i^{(n-1)}-z_j^{(n-1)}|+S_{n-1}(\bz_{n-1},\ldots \bz_1))
			\end{split}\eeqq
			It was also shown in \cite{BDKK2} that
			\beqq T_n(\bx_n,\bz_{n-1}, \bz_{n-2},\ldots, \bz_1)<C(\bx_n)- \ve \sum_{i=1}^{n-1}\sum_{j=1}^i |z_j^{(i)}|\eeqq	  
				Thus we can bound the integrand of \rf{Q37} by 
				\beq \label{Q43} C(\bx_n)\exp(-\ve(\sum |y_j| +\sum |z_j^{k}|)+(2g-1)|z_{n-1}^{(n-1)}-y_{n-1}|)\eeq
				This proves the convergence of the integral \rf{Q37} for $0<\Re g<1/2$ and thus the equality \rf{Q35} for $0<\Re g<1/2$. However, both sides of equality make sense for $0<\Re g<1$ and depend analytically on $g$. So we can freely extend the statement to the region $0<\Re g<1$.	  
				\end{proof}
			\begin{corollary} For $0<g<1$ the operators 
			\beq\label{Q44} Q_{n}^*(\mu)=Q_{n,+}^*(\mu)+Q_{n,-}^*(\mu)\eeq
			form a commuting family of symmetric with respect to scalar product \rf{pair1} operators diagonalized by functions
			 $\Psi_{\bl_n}(\bx_n)$ with eigenvalues
			 \beq\label{Q45}
			 \prod_{j=1}^n\frac{\Gamma\left(\frac{1-g}{2}- i(\l_j-\mu)\right)}{\Gamma\left(\frac{1+g}{2}- i(\l_j-\mu)\right)} + \prod_{j=1}^n\frac{\Gamma\left(\frac{1-g}{2}+ i(\l_j-\mu)\right)}{\Gamma\left(\frac{1+g}{2}+ i(\l_j-\mu)\right)}.\eeq
			\end{corollary}
			\begin{proof} Due to \rf{pair3}, the operators $Q_{n}^*(\mu)$ are symmetric with respect to scalar product \rf{pair1} and are diagonal on the complete set of functions  $\Psi_{\bl_n}(\bx_n)$.
			
			\end{proof}

{\bf 4. Example $n=1$}

Compute directly the eigenvalue of $Q^*_{1,+}(\mu)$. We have
\beqq\begin{split} Q^*_{1,+}(\mu)(e^{2\pi i\gamma x})=&\frac{2\pi}{\Gamma(g)}\int\limits_{y>x} e^{2\pi i\mu(x-y)}\big(2\sh\pi(y-x)\big)^{g-1}e^{2\pi i\gamma y}dy
\\	\substack{ = \\  t=y-x,\\ \lambda=\mu-\gamma }\hspace{0.4cm}
	\frac{2\pi 	e^{2\pi i\gamma x}}{\Gamma(g)}&\int\limits_0^\infty e^{\pi(-2i\l+g-1)t}(1-e^{-2\pi t})^{g-1}dt	\\ \substack{ = \\  u=e^{-2\pi t} }\hspace{0.5 cm}&
		\frac{e^{2\pi i\gamma x}}{\Gamma(g)}\int\limits_0^1u^{\frac{1-g+2i\l}{2}-1}(1-u)^{g-1}du=
	\frac{\Gamma\left(\frac{1-g}{2}+i\lambda\right)}{\Gamma\left(\frac{1+g}{2}+i\lambda\right)}e^{2\pi i\gamma x}.
	\end{split}\eeqq
For real $\l=\mu-\gamma$ integral converges for $0<g<1$. 
Similarly, \beqq\begin{split}
 Q^*_{1,-}(\mu)(e^{2i\gamma x})=&\frac{2\pi}{\Gamma(g)}\int\limits_{y<x} e^{2\pi i\mu(x-y)}\big(2\sh\pi(x-y)\big)^{g-1}e^{2\pi i\gamma y}dy
 \\	\substack{ = \\  t=x-y,\\ \lambda=\mu-\gamma }\hspace{0.4cm}
 \frac{2\pi e^{2\pi i\gamma x}}{\Gamma(g)}&\int\limits_0^\infty e^{\pi(2i\l+g-1)t}(1-e^{-2\pi t})^{g-1}dt	\\ \substack{ = \\  u=e^{-2\pi t} }\hspace{0.5 cm}&
 \frac{e^{2\pi i\gamma x}}{\Gamma(g)}\int\limits_0^1u^{\frac{1-g-2i\l}{2}-1}(1-u)^{g-1}du=
 \frac{\Gamma\left(\frac{1-g}{2}-i\lambda\right)}{\Gamma\left(\frac{1+g}{2}-i\lambda\right)}e^{2\pi i\gamma x}.
	\end{split}\eeqq

\section{Limiting procedure}
We may obtain interlacing raising operator $\Lambda^*$ as a strict limit from reflected raising operator $\Lambda^{R*}$, corresponding to hyperbolic Ruijsenaars model. The latter was introduced and studied at \cite{BDKK4}. We suggest a way to obtain the limiting operator  providing required calculations.
 
The subtle  thing is that contrary to general hyperbolic system, the raising $\Lambda^*$ operator in the rational limit has additional integrable singularities along the diagonals,which prevent the direct application of dominated convergence theorem to pointwise limit of the operator's kernel, but requires instead additional bounds of integrals along the singularities.

To begin with, recall that reflected raising operator $\Lambda^{R*}$ was defined as
\beq\begin{split}\label{def_lr}
\Lambda^{R*}_n(\l)f(\bx_n) = d^{R*}_{n-1}(g;\omega_2)\int_{\mathbb{R}^{n-1}}\Lambda^{R*}(\bx_n;\by_{n-1};\l)f(\by_{n-1})d\by_{n-1},\\
\Lambda^{R*}(\bx_n;\by_{n-1};\l) = (\eta^R)^{-1}(\bx_n)e^{2\pi\iota\l(\underline{\bx}_n-\underline{\by}_{n-1})}K^{R*}(\bx_n;\by_{n-1})\Delta^R(\by_{n-1}),
\end{split}\eeq
where $d^R_n(g)$ is the normalization constant, $\Delta^{R}$ -- the Sklyanin measure, $\eta^R$ -- the inverse of its shifted by $g$ modification, and $K^{R*}$ is the kernel function with reflected coupling constant. To be specific,
\beq\begin{split}
d^{R*}_n(g;\omega_2) = \frac{S_2(\omega_2g)^{n}}{\sqrt{\omega_2}^nn!}, ~~~~~ \Delta^R(\by_{n-1}) = \prod_{1\leq i<j\leq n-1}S_2(\pm\iota(y_i-y_j)),\\
\eta^R(\bx_n) = \prod_{1\leq i<j\leq n}S_2(\pm\iota(x_i-x_j)+\omega_2g)^{-1},\\
K^{R*}(\bx_n,\by_{n-1}) = \prod^n_{i=1}\prod^{n-1}_{j=1}S_2(\pm\iota(x_i-y_j)+\omega_2g/2)^{-1}.
\end{split}\eeq
We mean by $S_2(z) = S_2(z|1,\omega_2)$ the two-periodic meromorphic function called the double sine function. We refer to \cite{BDKK1, BDKK2} to its definition and basic properties.

In the following, we assume that all variables are real unless it is directly said the opposite.

As it was mentioned, if we tend $\omega_2\to0^+$, then obtain as a strict limit the interlacing raising operator $\Lambda^*$. We split the procedure onto two statements: derive pointwise limit of integrands and deduce a bound that is uniform in $\omega_2$ by set of absolutely integrable functions.

Let us rewrite the definition of $\Lambda^*$ in a manner similar to \rf{def_lr}
\beqq\begin{split}
\Lambda^*_n(\l)f(\bx_n)=d^*_{n-1}\int_{x_1<y_1<\ldots<y_{n-1}<x_n}\Lambda^*(\bx_n;\by_{n-1};\l)f(\by_{n-1})d\by_{n-1},
\end{split}\eeqq
where
\beq\begin{split}
\Lambda^*(\bx_n;\by_{n-1};\l)= \prod_{1\leq i<j\leq n}|2\sh\pi(x_i-x_j)|^{-2g+1}& e^{2\pi\iota\l(
\underline{\bx_n}-\underline{\by}_{n-1})}\\
\prod_{1\leq i < j\leq n-1}|2\sh\pi(y_i-y_j)|&\prod^n_{i=1}\prod^{n-1}_{j=1}|2\sh\pi(x_i-y_j)|^{g-1}
\end{split}\eeq
and $d^*_n$ is taken from \rf{I3}.

\begin{theorem}\label{thm4.1}
For $g\in(0,+\infty)$ and $x_1<\ldots<x_n$, we have the limit
\beq\begin{split}
\lim_{\omega_2\to 0^+}\big(2\pi\omega_2\big)^{(n-1)g}d^R_{n-1}(1 + \omega_2(1-g))\Lambda^{R*}(\bx_n;\by_{n-1};\l)\\
= \frac{d^*_{n-1}}{(n-1)!}\Lambda^*(\bx_n;\by_{n-1};\l)\sum_{\sigma\in S_{n-1}}\mathrm{Ind}(x_1<y_{\sigma(1)}<\ldots<y_{\sigma(n-1)}<x_n).
\end{split}\eeq
\end{theorem}
 Here  $\text{Ind}(\mathrm{A})$  means the logical function equal to $1$ once the condition $\mathrm{A}$ is satisfied and zero otherwise. 
\begin{proof}
Main tool is the limit for $g > 0$
\beq\begin{split}\label{pnt_s2}
\lim_{\omega_2\to 0^+}\frac{S_2(\iota x)}{S_2(\iota x+\omega_2g)} = \big(2\iota\sh\pi x\big)^g, ~~~~~ x\in\mathbb{R},
\end{split}\eeq

Sklyanin measure is directly equal to
\beq\begin{split}\label{skl_msr}
\prod_{1\leq i<j\leq n-1}|2\sh\pi\omega^{-1}_2(y_i-y_j)||2\sh\pi(y_i-y_j)|,
\end{split}\eeq
therefore,
\beq\begin{split}\label{lim_del}
\lim_{\omega_2\to0^+}\exp\pi\omega^{-1}_2\big[-\sum_{1\leq i<j\leq n-1}|y_i-y_j|\big]\Delta^R(\by_{n-1}) = \prod_{1\leq i<j\leq n-1}|2\sh\pi(y_i-y_j)|.
\end{split}\eeq

Applying \rf{pnt_s2}, we immediately derive
\beq\begin{split}\label{lim_etk}
\lim_{\omega_2\to0^+}\exp\pi\omega^{-1}_2\big[\sum_{1\leq i<j\leq n}|x_i-x_j|\big]\eta^R(\bx_n) = \prod_{1\leq i<j\leq n}|2\sh\pi(x_i-x_j)|^{2g-1};\\
\lim_{\omega_2\to0^+}\exp\pi\omega^{-1}_2\big[\sum^n_{i=1}\sum^{n-1}_{j=1}|x_i-y_j|\big]K^{R*}(\bx_n;\by_{n-1}) = \prod^n_{i=1}\prod^{n-1}_{j=1}|2\sh\pi(x_i-y_j)|^{g-1}.
\end{split}\eeq

The remaining part is $d^R_{n-1}(1+\omega_2(1-g))$. We will use another general limit for $\Re z > 0$
\beq\begin{split}
\lim_{\omega_2\to0^+}\omega_2^{z-1/2}S_2(\omega_2z|1,\omega_2) = (2\pi)^{1-z}\Gamma(z)^{-1}.
\end{split}\eeq
It follows that
\beq\begin{split}
\lim_{\omega_2\to0^+}\omega^{(n-1)g}_2d^R_{n-1}(1+\omega_2(1-g)) = \frac{(2\pi)^{(n-1)(1-g)}\Gamma(g)^{-(n-1)}}{(n-1)!}.
\end{split}\eeq

Consider accumulated exponential terms from \rf{lim_del} and \rf{lim_etk} (do not forget the corresponding powers from $\Lambda^{R*}$)
\beq\begin{split}
\exp\pi\omega^{-1}_2\Big[\sum_{1\leq i<j\leq n}|x_i-x_j| - \sum^n_{i=1}\sum^{n-1}_{j=1}|x_i-y_j| + \sum_{1\leq i<j\leq n-1}|y_i-y_j|\Big].
\end{split}\eeq
It turns out that the exponent itself tends to sum of well-defined indicator function as $\omega_2\to0^+$. This follows from Rains lemma \cite[Lemma3.6]{R} (its partial case was used in \rf{Q34}).

\begin{lemma}
\label{lem4.1}
Assume that $x_1<\ldots<x_n$, $y_1<\ldots<y_{n-1}$, and all variables are pairwise distinct. Then
\beq\begin{split}
\sum_{1\leq i<j\leq n}|x_i-x_j| - \sum^n_{i=1}\sum^{n-1}_{j=1}|x_i-y_j| + \sum_{1\leq i<j\leq n-1}|y_i-y_j| \leq 0
\end{split}\eeq
with zero value exactly at region $x_1<y_1<\ldots<x_{n-1}<y_{n-1}<x_n$.
\end{lemma}
\color{black}
As a result, we obtain the desired limit
\beq\begin{split}
\lim_{\omega_2\to0^+}\big(2\pi\omega_2\big)^{(n-1)g}d^R_{n-1}(1+\omega_2(1-g))
\eta^R(\bx_n)^{-1}e^{2\pi\iota\l(\underline{\bx}_n-\underline{\by}_{n-1})}\\
K^{R*}(\bx_n;\by_{n-1})\Delta^R(\by_{n-1}) = \frac{d^*_{n-1}}{(n-1)!}\eta(\bx_n)^{-1}e^{2\pi\iota\l(\underline{\bx}_n-\underline{\by}_{n-1})}\\K^*(\bx_n;\by_{n-1})\Delta(\by_{n-1})\sum_{\sigma\in S_{n-1}}\mathrm{Ind}(x_1<y_{\sigma(1)}<\ldots< y_{\sigma(n-1)} <x_n).
\end{split}\eeq
\end{proof}

To make motivating pointwise limit from Theorem \ref{thm4.1} the strict one, we would like to provide a set of bounds for integrand $\Lambda^{R*}(\bx_n;\by_{n-1};\l)$ that happen to be absolutely integrable. Additionally, we have to show the absolute convergence of the limiting operator which was done at Proposition \ref{proposition1a} due to \rf{le_lst}. Thus, dominated convergence theorem proves the limit of operators over the space of bounded functions $f(\by_{n-1})$.

\begin{proposition}\label{prp4.1}
Fix sufficiently small $\epsilon > 0$ and assume that $x_i\neq x_j$ for $1\leq i<j\leq n$. One may show that
\beq\begin{split}\label{prp_bnd1}
\Big|\Lambda^{R*}(\bx_n;\by_{n-1};\l)\Big| \leq& C\Big|\Lambda^*(\bx_n;\by_{n-1};\l)\Big|\times\\
&\frac{\prod_{1\leq i<j\leq n}|\sh\pi\omega^{-1}_2(x_i-x_j)|\prod_{1\leq i<j\leq n-1}|\sh\pi\omega^{-1}_2(y_i-y_j)|}{\prod^n_{i=1}\prod^{n-1}_{j=1}|\sh\pi\omega^{-1}_2(x_i-y_j)|}
\end{split}\eeq
for $\by_{n-1}\in\mathbb{R}^{n-1}$ such that $y_i\neq y_j$ for $1\leq i<j\leq n-1$, and $\omega_2\in(0,\epsilon)$.
\end{proposition}
\begin{proof}
Use \cite[Proposition A.2, Proposition A.3]{BCDK} to deduce for $g>0$
\beq\begin{split}
\bigg|\frac{S_2(\iota x)}{S_2(\iota x + \omega_2g)}\bigg| \leq C|\sh\pi x|^g, ~~~~~ x\in\mathbb{R}, ~~~\omega_2\to0^+,
\end{split}\eeq
where $C$ does not depend on the choice of either $x$ or $\omega_2$.

Then, in measure $\Delta^R(\by_{n-1})$, we obtain directly what is needed (see \rf{skl_msr}). As for the kernel function $K^{R*}(\bx_n;\by_{n-1})$, we may rewrite all the terms in a manner
\beq\begin{split}
\frac{1}{\big|S_2(\pm\iota(x_i-y_j)+\omega_2g/2)\big|}=\qquad&\\
\Bigg|\frac{S_2(\pm\iota(x_i-y_j))}{S_2(\pm\iota(x_i-y_j)+\omega_2g/2)}\Bigg|&\frac{1}{|2\sh\pi\omega^{-1}_2(x_i-y_j)||2\sh\pi(x_i-y_j)|}\\
&\leq C|\sh\pi(x_i-y_j)|^{g-1}|\sh\pi\omega^{-1}_2(x_i-y_j)|^{-1}.
\end{split}\eeq
\end{proof}

We are going to show that the limiting kernel $\Lambda^*(\bx_n;\by_{n-1};\l)$ is absolutely integrable for any $g>0$. From the side of operator, it says that $\Lambda^*_n$ is defined over the space of bounded functions $f(\by_{n-1})$ along $[x_1,x_n]$.

\begin{proposition}\label{prp4.2}
Assume that $x_1<\ldots<x_n$ and $g\in(0, \infty)$. The function $\Lambda^*(\bx_n;\by_{n-1};\l)$ is absolutely integrable over $\by_{n-1}$ satisfying $x_1<y_1<\ldots<x_{n-1}<y_{n-1}<x_n$.
\end{proposition}

The proof is due to the Lemma \ref{le_lst}. We place $a = g-1$ and obtain the desired integrability, but for rational functions. The case of the hyperbolic ones has to be equivalent since we sit on a compact.

To deduce the operator limit $\Lambda^{R*}\mapsto\Lambda^*$, an extremely important statement is the following lemma, which strengthen the statement of Rains lemma. It helps to move from analyzing asymptotics of multidimensional integrals to bounds for 1-dimensional case.
\begin{lemma}\label{lem4.2}
Consider pairwise distinct variables $\bx_n$, $\by_{n-1}$. Assume that $x_1<\ldots<x_n$ and $y_1<\ldots<y_{n-1}$. Then there exist $\alpha(i) \in\mathbb{Z}$ and 
injection $f : \{1,\ldots,n-1\}\to\{1,\ldots,n\}$, such that
\beq\begin{split}
\sum_{1\leq i<j\leq n}|x_i-x_j| - \sum^n_{i=1}\sum^{n-1}_{j=1}|x_i-y_j|+\sum_{1\leq i<j\leq n-1}|y_i-y_j|
= \sum^{n-1}_{i=1}2\alpha(i)(y_i-x_{f(i)})
\end{split}\eeq
and
 $$\alpha(i)(y_i-x_{f(i)}) \leq 0 ~~~ \mathrm{for~all} ~~~ i=1,\ldots,n-1.$$
Moreover, if either $y_j<x_1$ or $x_n<y_j$, then $\alpha(j)(y_j-x_{f(j)}) < 0$.
\end{lemma}
\begin{proof}
The first part of statements was proven in \cite[Appendix B.2]{BCDK}. The second one has to follow from the following remark:

If $y_j<x_1$, then $\alpha(1) = 1$, $y_1-x_{f(1)}>0$, so $\alpha(1)(y_1-x_{f(1)}) < 0$. Similar holds true for $x_n<y_j$.
\end{proof}

\begin{theorem}\label{thm4.2}
Consider a bounded function $f(\by_{n-1})$ that is defined everywhere along $\by_{n-1}\in\mathbb{R}^{n-1}$. Then for $g > 1$ we establish the limit $\Lambda^{R*}_n(\l)[f](\bx_n) \mapsto\Lambda^*_n(\l)[f](\bx_n)$ as a limit of operators.
\end{theorem}
\begin{proof}
First of all, assume that $x_1<\ldots<x_n$.

According to Proposition \ref{prp4.2}, we deal with absolutely convergent integral in limit. Due to Proposition \ref{prp4.1}, we may consider the convergence of $RHS$ of \rf{prp_bnd1}. It splits onto two questions: singularities and asymptotics away from a compact.

To answer on them, consider one-sided neighbourhoods in vicinity of $x_1,\ldots,x_n$. Define them as
\beq\begin{split}
H^+_i := (x_i,x_i+\omega_2], \qquad H^-_i := [x_i-\omega_2,x_i). 
\end{split}\eeq
We are going to integrate over all possible orderings of the tuple $(x_1,\ldots,x_n,y_1,\ldots,y_{n-1})$ distinguishing sitting inside a neighbourhood and staying outside of it.

Assume that $y_1<\ldots<y_{n-1}$. The strategy is to take 1-dimensional integrals one-by-one with respect to picked ordering.

\pmb{0.} Preparation. First things first, we rewrite $\omega_2$-dependent terms as a combination of asymptotical exponent and the rest ones that can still have singularities. Then we get rid of resulting measure parts since their powers are positive. Namely,
\beq\begin{split}
\prod_{1\leq i<j\leq n}|\sh\pi\omega^{-1}_2(x_i-x_j)|\prod^n_{i=1}\prod^{n-1}_{j=1}|\sh\pi\omega^{-1}_2(x_i-y_j)|^{-1}\prod_{1\leq i<j\leq n-1}|\sh\pi\omega^{-1}_2(y_i-y_j)|\\
\leq \exp\pi\omega^{-1}_2\Big(\sum_{1\leq i<j\leq n}|x_i-x_j|-\sum^n_{i=1}\sum^{n-1}_{j=1}|x_i-y_j|+\sum_{1\leq i<j\leq n-1}|y_i-y_j|\Big)\\
\prod^n_{i=1}\prod^{n-1}_{j=1}\big(1-e^{-2\pi\omega^{-1}_2|x_i-y_j|}\big)^{-1} =: \exp\pi\omega^{-1}_2A(\bx_n;\by_{n-1}) \cdot B(\bx_n;\by_{n-1}).
\end{split}\eeq

Instead of $\big|\Lambda^*(\bx_n;\by_{n-1};\l)f(\by_{n-1})\big|$, we consider
\beq\begin{split}
\exp\Big(c\sum_{1\leq i\leq n-1}|y_i|\Big)\prod^n_{i=1}\prod^{n-1}_{j=1}\big(1-e^{-2\pi|x_i-y_j|}\big)^{g-1} =: C(\bx_n;\by_{n-1})\cdot D(\bx_n;\by_{n-1}).
\end{split}\eeq

\pmb{1.} Integration over neighbourhoods $H_j^+$, $H^-_j$. Suppose, we would like to integrate\\ $y_i < y_{i+1}<\ldots<y_{i+k-1}$ over the neighbourhood $H^+_j$ (for $H^-_j$, everything is similar). Bound everything in $C$ by corresponding boundary points. In $D$, we can omit all terms except for\\ $\prod^{k-1}_{l=0}\big(1-e^{-2\pi|y_{i+l}-x_j|}\big)^{g-1}$. The same is true for $B$. $A$ is non-positive at all. Therefore, it remains to integrate along the following strata configuration (for $g > 1$)
\beq\begin{split}
\int^{x_j+\omega_2}_{x_j}\int^{x_j+\omega_2}_{y_i}\ldots\int^{x_j+\omega_2}_{y_{i+k-2}}\prod^{k-1}_{l=0}\frac{\big(1-e^{-2\pi|y_{i+l}-x_j|}\big)^{g-1}}{1-e^{-2\pi\omega^{-1}_2|y_{i+l}-x_j|}}dy_i\ldots dy_{i+k-1}\\
\leq E\int^{x_j+\omega_2}_{x_j}\int^{x_j+\omega_2}_{y_i}\ldots\int^{x_j+\omega_2}_{y_{i+k-2}}\prod^{k-1}_{l=0}\frac{|y_{i+l}-x_j|^{g-1}}{\omega^{-1}_2|y_{i+l}-x_j|}dy_i\ldots dy_{i+k-1} \leq E'\omega_2^{kg}
\end{split}\eeq
for non-depending on $x_j$ and $\omega_2$ constants $E,E'$.

\pmb{2.} Integration over finite interval staying away from neighbourhoods. Bound $A,B,C,D$ putting corresponding boundary points. There is no problem with $1-e^{2\pi\omega^{-1}_2|y_i-x_j|}$ since minimal possible value of $\omega^{-1}_2|y_i-x_j|$ is 1.

\pmb{3.} Integration over infinite interval staying away from neighbourhoods. Suppose, we integrate along $x_n + 1 < y_i<\ldots<y_{n-1}$. Bound $B$ and $D$. Do the same for all terms of $A$ and $C$ except for $$\exp c|y_{n-1}|~~~\mathrm{and}~~~ \exp\pi\omega^{-1}_2\alpha(n-1)(y_{n-1}-x_{f(n-1)}),$$ where by Lemma \ref{lem4.2} $\alpha(i)(y_i-x_{f(i)}) < 0$ along the integration domain. Since $\omega_2\to 0^+$, we have exponential decreasing for sufficiently small $\omega_2$. Surely, resulting exponent term will contain $\omega^{-1}_2$ with negative coefficient. Namely,
\beq\begin{split}
\int^{+\infty}_{x_n+1}\exp\big(c|y_{n-1}|\big)\exp\pi\omega^{-1}_2\alpha(n-1)(y_{n-1}-x_{f(n-1)})dy_{n-1}\\
= \frac{\exp(c x_{f(n-1)})\exp(\pi\omega^{-1}_2\alpha(n-1)+c)(x_n+1-x_{f(n-1)})}{-\pi\omega^{-1}_2\alpha(n-1)-c} \to 0, ~~~ \omega_2\to 0^+.
\end{split}\eeq
with assumption that $x_n+1\geq 0$. It is important that $c$ is always can be chosen non-depending on $\omega_2$ which is an easy to follow along the proof restriction.

Since described three cases form self-similar picture, we are done.
\end{proof}

\section{\large Discussion}
\setcounter{equation}{0}

As well as in the general hyperbolic model, the family of raising operators $\Lambda^*_n(\mu)$ and commutative family $Q^*_n(\mu)$ of Baxter operators, defined and studied in this paper, deals with reflected coupling constant $1-g$ despite that the corresponding  well defined duality is missing in hyperbolic Sutherland model.

 We showed that for $g>1$ the raising operators $\Lambda^*_n(\mu)$ are actually the limits of the corresponding raising operators in general hyperbolic model. Presumably, the same holds for $0<g<1$ as well. 
 
 Unfortunately, this is not true for the Baxter operators $Q^*_n(\mu)$; the limit from general hyperbolic system diverges, and it is not clear how do they arise. However,  
 we are able to construct in a lot more sophisticated way similar limits with compositions of either raising operators or common-sign Baxter operators. Precisely,
\beq\begin{split}
\Lambda^*_{n+1}(\l)\Lambda^*_n(\mu)[f](\bx_{n+1}) = \Lambda^*_{n+1}(\mu)\Lambda^*_n(\l)[f](\bx_{n+1});\\
Q^*_{n,+}(\l)Q^*_{n,+}(\mu)[f](\bx_n) = Q^*_{n,+}(\mu)Q^*_{n,+}(\l)[f](\bx_n);\\
Q^*_{n,-}(\l)Q^*_{n,-}(\mu)[f](\bx_n) = Q^*_{n,-}(\mu)Q^*_{n,-}(\l)[f](\bx_n).
\end{split}\eeq

They are considered from the side of the kernel functions for the first integration out of two, and follow from the hyperbolic level. Meanwhile, the limits are proved only for $g\in(1,\frac{3}{2})$ since we can take uniformal bounds of the ratio of double sine functions in distinct ways, so extra singularities may either arise at once or sit at the complex domain with non-zero imaginary part and just tend to the real line. Surely, it produces extra issues in construction of uniformal majorating function.
Actually, restriction $g\in(1,\frac{3}{2})$ is not a problem since resulting functions are analytic in $g$ with $\Re g > 0$.

The hypebolic Ruijsenaars model posseses one more pair of Baxter and raising operators, acting in the space of spectral parameters and related to the dual coupling constant. The main tool to find and study analogous operators for the Sutherland model    is the use of degenerate Rains integral identities. Such trigonometric degenerations into identities on classical Mellin - Barnes integrals, suitable for the dual Sutherland model, were recently described in \cite{Ch}. This is the subject of the forthcoming paper.

\section*{Acknowlegements} The main ideas of this work emerged as a result of numerous discussions and debates with N. Belousov and S.Derkachev.  We sincerely thank them for all their help.
 This work was done within research projects implemented as part of the Basic Research Program at the HSE University, HSE-BR-2025-84,  HSE-BR-2025-062.
\color{black}

\appendix
\section{Proof of Lemma \ref{le_lst}.}
{\bf Lemma 3.1}\
{\it For $a \in (-1,+\infty)$, we have the bound}
\beq\begin{split}
\int\limits_{x_1<y_1<\ldots<x_n}d\by_{n-1}\prod^n_{i=1}\prod^{n-1}_{j=1}|x_i-y_j|^a\times\prod_{1\leq i<j<n-1}|y_i-y_j| \leq C\prod_{1\leq i < j \leq n}|x_i-x_j|^{2a+1}.
\end{split}\eeq
The constant can be taken
\beq\begin{split}
	C = \begin{cases}
		1, \qquad a\in[0,+\infty),\\
		\frac{4^{-a(n-1)}}{(a+1)^{n-1}}, \qquad a\in[-\frac{1}{2}, 0],\\
		\frac{2^{-an(n-1)}}{\sqrt{2}^{(n-1)(n-2)}(a+1)^{n-1}}, \qquad a\in(-1,-\frac{1}{2}].
	\end{cases}
\end{split}\eeq

Actually, this lemma consists of three distinct ones. Namely, we distinguish cases with $a\in[0,+\infty)$, $a \in (-\frac{1}{2}, 0)$, and $a\in(-1,-\frac{1}{2}]$.

\pmb{a.} If $a\geq 0$, we can majorate all elementary terms in a manner of \textit{decoupling}. Precisely, 
\beq\begin{split}
\text{if} ~~ x_i \leq x_j < y_j < x_{j+1}, ~~ \text{then} ~~ |x_i-y_j|^a \leq |x_i-x_{j+1}|^a;\\
\text{if} ~~ x_j < y_j < x_{j+1} \leq x_i, ~~ \text{then} ~~ |x_i-y_j|^a \leq |x_i-x_j|^a.
\end{split}\eeq
\beq\begin{split}
\text{If} ~~ x_i < y_i < x_{i+1} \leq x_j < y_j < x_{j+1}, ~~ \text{then} ~~ |y_i-y_j| \leq |x_i-x_{j+1}|.
\end{split}\eeq
Immediately, we obtain
\beq\begin{split}
\prod^n_{i=1}\prod^{n-1}_{j=1}|x_i-y_j|^a \leq \prod_{1\leq i < j\leq n}|x_i-x_j|^{2a},
\end{split}\eeq
\beq\begin{split}
\prod_{1\leq i < j\leq n-1}|y_i-y_j| \leq \prod_{1\leq i < j\leq n}|x_i-x_j|\times\prod_{1\leq i\leq n-1}|x_i-x_{i+1}|^{-1}.
\end{split}\eeq
Then integration produces $\int_{x_1<y_1<\ldots<x_n}d\by_{n-1}=\prod_{1\leq i\leq n-1}|x_i-x_{i+1}|$. So we can take $C = 1$ for $a\in[0,+\infty)$.

\pmb{b.} If $a\in(-\frac{1}{2},0)$, things become more complicated. Driving ideas are the following. We look at measure and \textit{decouple} all terms in both directions separately, i.e. 
\beq\begin{split}
	|y_i-y_j| \leq |x_i-y_j|^{1/2}|y_i-x_{j+1}|^{1/2}, \qquad i < j.
\end{split}\eeq
Then we deal with two essentially distinct terms
\beq\begin{split}
	|x_i-y_j|^{1/2+a}, \qquad i-j \neq 0,1, \qquad \text{and} \qquad |x_i-y_i|^a|y_i-x_{i+1}|^a.
\end{split}\eeq
The first ones can be \textit{decoupled} again.

To remain consistent, write directly what happens with integrand
\beq\begin{split}
\prod^n_{i=1}\prod^{n-1}_{j=1}|x_i-y_j|^a\times\prod_{1\leq i<j\leq n-1}|y_i-y_j| \leq\\ \prod_{1\leq i \leq n-1}|x_i-y_i|^a|x_{i+1}-y_i|^a\times\prod^n_{i=1}\prod_{\substack{1\leq j\leq n-1\\i-j\neq 0,1}}|x_i-y_j|^{1/2+a} \leq\\ \prod_{1\leq i \leq n-1}|x_i-y_i|^a|x_{i+1}-y_i|^a\times\prod_{1\leq i<j\leq n}|x_i-x_j|^{1+2a}\times\prod_{1\leq i\leq n-1}|x_i-x_{i+1}|^{-1-2a}.
\end{split}\eeq
Thus, we want to calculate the integral
\beq\begin{split}\label{int_a}
\int_{x_i<y_i<x_{i+1}}dy_i|x_i-y_i|^a|x_{i+1}-y_i|^a \leq \int^{(x_i+x_{i+1})/2}_{x_i}dy_i|x_i-y_i|^a|x_i-x_{i+1}|^a2^{-a}\\
+\int^{x_{i+1}}_{(x_i+x_{i+1})/2}dy_i|x_i-x_{i+1}|^a2^{-a}|x_{i+1}-y_i|^a = \frac{2^{-2a}}{a+1}|x_i-x_{i+1}|^{1+2a}.
\end{split}\eeq
It ends the proof for $a\in(-\frac{1}{2},0)$. Note that we can take $C = \Big(\frac{2^{-2a}}{a+1}\Big)^{n-1}$ here.

\pmb{c.} For $a\in(-1,-\frac{1}{2})$, we remark that \textit{decoupling} squeezes the interval instead of unclenching it (since the power is negative). By definition,
\beq\begin{split}
\text{if} \qquad x_i < x_j < y_j < x_{j+1}, \qquad \text{then} \qquad |x_i-y_j|^a \leq |x_i-x_j|^a;\\
\text{if} \qquad x_j < y_j < x_{j+1} < x_i, \qquad \text{then} \qquad |x_i-y_j|^a \leq |x_i-x_{j+1}|^a.
\end{split}\eeq
Note that we cannot \textit{decouple} $|x_i-y_i|^a$ and $|x_{i+1}-y_i|^a$, evidently.

As it was done in proof of the case $a\in(-\frac{1}{2},0)$, we want to replace measure by some terms from kernel function, then \textit{decouple} everything we can (depends on power) and justify the statement from lemma.

Consider $|y_i-y_j|$ with $i < j$. Note that
\beq\begin{split}
	|y_i-y_j| = |y_i-y_j|^{\alpha}|y_i-y_j|^{1-\alpha} \leq |x_i-y_j|^{\alpha}|y_i-x_{j+1}|^{1-\alpha}, \qquad \alpha \in [0, 1].
\end{split}\eeq
Assume that $\max\{|x_i-x_j|, |x_{i+1}-x_{j+1}|\} = |x_i-x_j|$, then take $\alpha = 1+a$ and \textit{decouple}
\beq\begin{split}\label{a_ineq}
	|y_i-y_j| \leq |x_i-y_j|^{1+a}|y_i-x_{j+1}|^{-a}.
\end{split}\eeq	
Otherwise, we take $\alpha = -a$.

Since any two measure terms affect exactly four distinct kernel terms, it is straightforward to deduce the resulting powers. Indeed, applying (\ref{a_ineq}), we arrive to inequality
\beq\begin{split}
\prod^n_{i=1}\prod^{n-1}_{j=1}|x_i-y_j|^a\times\prod_{1\leq i<j\leq n-1}|y_i-y_j| \leq&\\ \prod_{1\leq i\leq n-1}|x_i-y_i|^a&|x_{i+1}-y_i|^a \prod_{1\leq i<j\leq n-1}p(x_i,y_i,y_j,x_{j+1})
\end{split}\eeq
with
\beq\begin{split}
p(x_i,y_i,y_j,x_{j+1}) = \begin{cases}
|x_i-y_j|^{1+2a}, \qquad |x_i-x_{i+1}| \geq |x_j-x_{j+1}|,\\
|y_i-x_{j+1}|^{1+2a}, \qquad |x_j-x_{j+1}| \geq |x_i-x_{i+1}|.
\end{cases}
\end{split}\eeq
Since $1+2a \leq 0$, we will \textit{decouple} all $p(x_i,y_i,y_j,x_{j+1})$ to $q(x_i,x_{i+1},x_j,x_{j+1})$. Namely, $p(x_i,y_i,y_j,x_{j+1}) \leq q(x_i,x_{i+1},x_j,x_{j+1})$ with
\beq\begin{split}
q(x_i,x_{i+1},x_j,x_{j+1}) := \begin{cases}
|x_i-x_j|^{1+2a}, ~~~ |x_i-x_{i+1}|\geq|x_j-x_{j+1}|,\\
|x_{i+1}-x_{j+1}|^{1+2a}, ~~~ |x_j-x_{j+1}|\geq|x_i-x_{i+1}|.
\end{cases}
\end{split}\eeq
For simplicity, rewrite
\beq\begin{split}
q(x_i,x_{i+1},x_j,x_{j+1}) =: |x_i-x_j|^{1/2+a}|x_{i+1}-x_{j+1}|^{1/2+a}\bigg(\frac{|x_i-x_j|^{1/2+a}}{|x_{i+1}-x_{j+1}|^{1/2+a}}\bigg)^{\sigma_{ij}}
\end{split}\eeq
with $\sigma_{ij} := \text{sign}(|x_{i+1}-x_i|-|x_{j+1}-x_j|)$.

Finally, we integrate the remaining part (see (\ref{int_a})) as
\begin{equation*}\begin{split}
\int_{x_1<y_1<\ldots<y_{n-1}<x_n}d\by_{n-1}\prod^{n-1}_{i=1}|y_i-x_i|^a|x_{i+1}-y_i|^a = \Big(\frac{2^{-2a}}{a+1}\Big)^{n-1}\prod^{n-1}_{i=1}|x_i-x_{i+1}|^{1+2a}.
\end{split}\end{equation*} 

Therefore we bound the whole integral from lemma by
\beq\begin{split}
\Big(\frac{2^{-2a}}{a+1}\Big)^{n-1}\prod_{1\leq i < j\leq n}|x_i-x_j|^{1+2a}\times\frac{\prod_{1\leq i\leq n-1}|x_i-x_{i+1}|^{1+2a}}{\prod_{1\leq i\leq n-1}|x_1-x_{i+1}|^{1/2+a}|x_i-x_n|^{1/2+a}}\\
\prod_{1\leq i < j\leq n-1}\bigg(\frac{|x_i-x_j|^{1/2+a}}{|x_{i+1}-x_{j+1}|^{1/2+a}}\bigg)^{\sigma_{ij}}.
\end{split}\eeq

It remains to show that error term is bounded by something non-depending on parameters $\bx_n$. Reminding that $\frac{1}{2} + a \leq 0$, we end the proof by the following lemma.

\begin{lemma}\label{le_c}
We have the bound
\beq\begin{split}
\prod_{1\leq i < j\leq n-1}\bigg(\frac{|x_{i+1}-x_{j+1}|}{|x_i-x_j|}\bigg)^{\sigma_{ij}}\frac{\prod_{1\leq i\leq n-1}|x_1-x_{i+1}||x_i-x_n|}{\prod_{1\leq i\leq n-1}|x_i-x_{i+1}|^2} \leq C'
\end{split}\eeq
over $x_1 < x_2<\ldots<x_n$. We can take $C' = 4^{\frac{(n-1)(n-2)}{2}}$.
\end{lemma}
\begin{proof}
Prove it by induction in $n$. Base with $n = 1,2$ is trivial. Assume that we know the fact for $n-1$. Then
\beq\begin{split}
&\prod_{1\leq i < j\leq n-1}\bigg(\frac{|x_{i+1}-x_{j+1}|}{|x_i-x_j|}\bigg)^{\sigma_{ij}}\frac{\prod_{1\leq i\leq n-1}|x_1-x_{i+1}||x_i-x_n|}{\prod_{1\leq i\leq n-1}|x_i-x_{i+1}|^2} =\\
&\qquad\qquad\prod_{1\leq i < j\leq n-2}\bigg(\frac{|x_{i+1}-x_{j+1}|}{|x_i-x_j|}\bigg)^{\sigma_{ij}}\frac{\prod_{1\leq i\leq n-2}|x_1-x_{i+1}||x_i-x_{n-1}|}{\prod_{1\leq i\leq n-2}|x_i-x_{i+1}|^2}\\
&\qquad\qquad\qquad\qquad\prod_{1\leq i\leq n-2}\bigg(\frac{|x_{i+1}-x_n|}{|x_i-x_{n-1}|}\bigg)^{\sigma_{i,n-1}}\frac{|x_1-x_n|^2}{|x_{n-1}-x_n|^2}\frac{\prod_{1\leq i\leq n-2}|x_{i+1}-x_n|}{\prod_{1\leq i\leq n-2}|x_i-x_{n-1}|},
\end{split}\eeq
so we have to bound the last row.

Here, we introduce another induction. To do it, define
\beq\begin{split}
M_k(x_k,\ldots,x_n) := \frac{|x_k-x_n|^2}{|x_{n-1}-x_n|^2}\prod_{k\leq i\leq n-2}\bigg(\frac{|x_{i+1}-x_n|}{|x_i-x_{n-1}|}\bigg)^{1+\sigma_{i,n-1}}.
\end{split}\eeq
Note that, initially, we were interested in $M_1(\bx_n)$.\\
We claim that
\beq\begin{split}
	M_k(x_k,\ldots,x_n) \leq C'', \qquad C'' = 4^{n-1-k}.
\end{split}\eeq
Base with $k = n-1$ is trivial. We take $C'' = 1$. Assume that we know the fact for $k$. Then
\beq\begin{split}
M_{k-1}(x_{k-1},\ldots,x_n) = M_k(x_k,\ldots,x_n)\frac{|x_{k-1}-x_n|^2}{|x_k-x_n|^2}\bigg(\frac{|x_k-x_n|}{|x_{k-1}-x_{n-1}|}\bigg)^{1+\sigma_{k-1,n-1}}.
\end{split}\eeq
Either $|x_{n-1}-x_n|\geq|x_{k-1}-x_k|$ and remaining term is 
\beq\begin{split}
\frac{|x_{k-1}-x_n|^2}{|x_k-x_n|^2} = 1 + 2\frac{|x_{k-1}-x_k|}{|x_k-x_n|} + \frac{|x_{k-1}-x_k|^2}{|x_k-x_n|^2} \leq 4,
\end{split}\eeq
or $|x_{n-1}-x_n|\leq|x_{k-1}-x_k|$ and remaining term is 
\beq\begin{split}
\frac{|x_{k-1}-x_n|^2}{|x_{k-1}-x_{n-1}|^2} = 1 + 2\frac{|x_{n-1}-x_n|}{|x_{k-1}-x_{n-1}|} + \frac{|x_{n-1}-x_n|^2}{|x_{k-1}-x_{n-1}|^2} \leq 4.
\end{split}\eeq

As a consequence, we can take $C' = 4^{\frac{(n-1)(n-2)}{2}}$.
\end{proof}
If we raise inequality from Lemma \ref{le_c} to non-negative power of $-a-\frac{1}{2}$, then we end the proof for $a\in(-1,-\frac{1}{2}]$. Resulting $C$ equals to
\beq\begin{split}
	\frac{2^{-an(n-1)}}{\sqrt{2}^{(n-1)(n-2)}(a+1)^{n-1}}.
\end{split}\eeq

	\end{document}